\documentclass[english]{article}

\usepackage{graphicx,latexsym,verbatim,amsfonts,amsmath,amssymb,amsmath, geometry}

\usepackage{algorithmic, graphics}
\usepackage[Algorithm]{algorithm}
\usepackage[latin1]{inputenc}
\usepackage[T1]{fontenc}

\newtheorem{definition}{Definition}
\newtheorem{theorem}{Theorem}
\newtheorem{example}{Example}
\newtheorem{property}{Property}
\newtheorem{corollary}{Corollary}
\newtheorem{lemma}{Lemma}

\newcommand{\rcp}{\longleftarrow}
\newenvironment{proof}{{\bf Proof }}{\begin{flushright}$\square$\end{flushright}}

\title{A Distributed Clustering Algorithm for Dynamic Networks}
\author{Thibault Bernard${}^{(1)}$, Alain Bui${}^{(2)}$, Laurence Pilard${}^{(2)}$ and Devan Sohier${}^{(2)}$}

\date{${}^{(1)}$ CReSTIC, Université de Reims CA, France, \\
{\small thibault.bernard@univ-reims.fr } \\
${}^{(2)}$
PR\emph{i}SM UMR CNRS, Université de Versailles Saint-Quentin, France,\\
{\small \{alain.bui, laurence.pilard, devan.sohier\}@prism.uvsq.fr}
}
\begin{document}

\maketitle

\begin{abstract}
We propose an algorithm that builds and maintains clusters over a network subject to mobility. This algorithm is fully decentralized and makes all the different clusters grow concurrently. The algorithm uses circulating tokens that collect data  and move according to a random walk traversal scheme. Their task consists in (i) creating a cluster with the nodes it discovers and (ii) managing the cluster expansion; all decisions affecting the cluster are taken only by a node that owns the token. The size of each cluster is maintained higher than $m$ nodes ($m$ is a parameter of the algorithm). The obtained clustering is locally optimal in the sense that, with only a local view of each clusters, it computes the largest possible number of clusters (\emph{ie} the sizes of the clusters are as close to $m$ as possible). This algorithm is designed as a decentralized control algorithm for large scale networks and is mobility-adaptive: after a series of topological changes, the algorithm converges to a clustering. This recomputation only affects nodes in clusters in which topological changes happened, and in adjacent clusters.
\end{abstract}

\section{Introduction}

Scalability in distributed system  has become a major challenge nowadays, in structuring and managing communications. We propose a solution to manage large-scale networks based on the division of the system into subsystems, called clusters. We focus in this paper on algorithms that build clusters and maintain them after topological reconfiguration. The algorithms we propose are decentralized: all nodes execute the same code. This allows all clusters to be built concurrently, which is desirable for efficiency.

Large-scale networks are often subject to mobility: their components can connect or disconnect. This phenomenon has to be taken into account. The algorithm being decentralized also allows the algorithm to have no distinguished node, the failure of which would lead to a major re-clustering. The connection or disconnection of a node has only a limited impact (that we can state) on the algorithm.

Random walks are naturally adaptive to dynamic networks such as ad-hoc sensors network \cite{BBCD02,DoSW06} because they make use only of local up-to-date information. Moreover they can easily manage connections and disconnections occurring in the network.
%le critere de decentralisation
%To fulfill the criteria of dynamicity, a clustering solution, must be totally decentralized: no node can be a leader , since the mobility could evict them from the network. Moreover, a reconfiguration of the network due to mobility should never entail a reconstruction of all the clusters.
%Moreover, after a single link failure for example, it is desirable that only a limited portion of the system needs to be reconfigured.

%% notre contexte de travail
Our solution takes these different constraints into account. It is based on the circulation of several tokens. Each token creates a cluster and coordinates its growing in a decentralized way.

A random walk based algorithm is a token circulation algorithm in which a token randomly moves among the nodes in the network. A random walk can be used as base of a distributed token circulation algorithm. This token collects and disseminates information in the network. At each step of the execution of the algorithm, the random walk (the token) is on a node $i$ of the network. The node that owns the token chooses one of its neigbour $j$ with a probability $1/degree(i)$. %(and more generally in the case of weighted graphs $\omega(i,j)/ \omega(i)$ where $\omega(i,j)$ is the weight of the edge $(i,j)$ and $\omega(i)$ is the sum of weights of all edges adjacent to $i$. 
It is important to remark that this definition ensures that all nodes, with high probability, %Puisqu'on est plus dans le cadre dynamique
 eventually own the token, and that the token, with high probability, %meme raison
 eventually hits all nodes \cite{Lova93}.

In \cite{BeBF04b,BBFR06b}, we introduced and used the combination of a circulating word, \emph{i.e.} the token has a content to collect and broadcast data (this concept is formally defined  in Section \ref{sectionCirculatingWord}) and a  random walk as moving scheme of the token. Using this combination, we proposed solutions to build adaptive spanning trees for systems like ad-hoc sensor networks. These solutions are tolerant to transient failures in the network. 

In these works, we also proposed a fully decentralized solution to the communication deadlock problem, introducing a new control mechanism called \emph{reloading wave}. These works have been used to propose a solution to the resource allocation problem in ad-hoc networks \cite{BBFN10}. Such a combination has also been used in \cite{BBFR06b} to build and maintain spanning structures to solve on-the-fly resources research in peer-to-peer grids.

Although the token perpetually circulates in the network in order to update the underlying structure, we bound the size of the circulating word to $2n-1$ in the case of bidirectional communication links and to $n^2/4$ in the case of unidirectional communication links, by retaining only the most recent data necessary to build the tree (with $n$ the size of the network, \cite{Bern06}). 

We use the content of the circulating word to structure the network into different clusters. Their construction and their maintenance are achieved in a decentralized way. Using the properties of random walks and of circulating words, the clusters are able to adapt to topological reconfigurations. Thus, this solution can be used to design distributed control algorithm on large scale dynamic networks.

%% Positionnement de notre solution par rapport aux autres
Unlike solutions described in \cite{Basa99,JoNg06}, our solution does not use any local leader on a cluster. The advantage of such solutions is that if a node ``moves'' in the network, this never entails a total reconstruction of the clusters. After a topological change, the system eventually converges to a correct global state without having to rebuild all clusters. This kind of approach on a 1-hop solution is described in \cite{FIMT05} in which re-clustering mechanism are used. Our solution is totally decentralized as opposed to \cite{BBCD02}, in which a spanning structure of the whole network is built in a first step, to be divided using a global mechanism. Our solution is realized in a fully concurrent way. As stated  in \cite{ABCP96}, it considerably accelerates the construction of the different clusters. Thus our solution satisfies the property highlighted in \cite{FaLa08}.

Moreover, we guarantee that after a topological change, only a bounded portion of the system is affected. Nodes that are in clusters that are not adjacent to the one in which it occurs have no extra work, and are not even aware of this event.

In the first section, we present some preliminary notions about random walk based distributed algorithms, and we present with more details the clustering problem we solve. The second section gives the fundamental distributed clustering algorithm we designed. The third section provides proofs about the correctness of the algorithm. The fourth section is about mobility: we present the slight adaptations needed to handle nodes and links mobility, as well as proofs of this algorithm; we also present a locality result: after topological modifications, in the worst case, the only clusters that are affected are the clusters in which the topological changes took place and clusters that are adjacent to them. Finally, we conclude this paper by presenting some future works.

\section{Preliminaries}

In this section, we define a distributed system, a random walk and a circulating word. We also introduce the notion of clustering.
 
\subsection{Distributed system} 

A distributed system is a connected graph $G =(V,E)$, where $V$ is a set of nodes with $|V| = n$ and $E$ is the set of bidirectional communication links. A node is constituted of a computing unit and a message queue. Two nodes are called neighbors if a communication link $(i,j)$ exists between them. Every node $i$ can distinguish between all its communication links and maintains a set of its neighbors (denoted $N_i$). The degree of $i$ is the number of its neighbors. We consider distributed systems in which all nodes have distinct identities. 

We consider asynchronous systems: we do not assume any bound on communication delays and on processing times.

\subsection{Random walk}

A random walk is a sequence of nodes visited by a token that starts at $i$ and visits other nodes according to the following transition rule: if the token is at $i$ at step $t$ then at step $t+1,$ it will be at one of the neighbors of $i$, this neighbor being chosen uniformly at random among all of them  \cite{Lova93,AKLL+79}. %As with deterministic distributed algorithms, the time complexity of a random walk based algorithm is the number of steps or messages it takes for the algorithm to achieve its goal.

A random walk on a connected graph eventually visits any node with probability 1 (\emph{whp} in the following). It means that, for any finite time $t$, it is possible that the walk does not hit a given node for $t$ steps, but that the probability that it does so tends to 0 as $t$ grows (\cite{Lova93}).

%Both the \emph{cover time} $C$ --- the average time to visit all nodes in the system --- and the \emph{hitting time} denoted by $h(i,j)$ --- the average time to reach a node $j$ for the first time starting from a given node $i$ --- are important values that appear in the analysis of random walk-based distributed algorithms. 
%Original method to efficiently compute hitting times and cover time is described in \cite{BBBS03,BuSo07}.

\subsection{Circulating word} \label{sectionCirculatingWord}

A circulating word is a list of data gathered during the circulation of the token in the network. In this work, the gathered data is the identifiers of the visited nodes. The word is represented as follow: $w = {<}w[1],\ldots, w[k]{>}$ where $k$ is the size of the circulating word and $w[l]$ is the node identifier in position $l$ in the word. Each time the token visits a node, the node identifier is added at the beginning of the word. 

The token follows a random walk scheme, which allows an efficient management of network. The word gathers identifiers of the visited nodes in order to maintain an adaptive spanning tree of the network. Only the most recent data is used to build this spanning tree, and, thus, only $2n-1$ entries of the word are used. Older data can be removed, which bounds the word length to $2n-1$. The detailed procedure to reduce the size of the word is described in \cite{Bern06} (cf. algorithm \ref{algo-cleanWord}).

\noindent We use the following procedures:
\begin{itemize}
\setlength{\itemsep}{-0.05cm}
\item Procedure $Size(w: word): integer$ -- returns the size of the word $w$;
\item Procedure $Nb\_Identities(w: word): integer$ -- returns the number of distinct identities in the word $w$;
\item Procedure $Add\_Begin(j: identifier,w: word): word$ -- adds identifier $j$ at the beginning of $w$.
\end{itemize}

\noindent The procedure $Build\_Tree$ computes from a word $w$ a tree $\cal A$ rooted in $w[1]$.
\begin{algorithm}[H]
{\small
\begin{algorithmic}[1]
\STATE ${\cal A}\rcp \emptyset$
\STATE $Set\_Root({\cal A},w[1])$
\FOR{$k=2$ to $Size(w)$}
	\IF{$w[k]\not\in \cal A$}
		\STATE $add\_Tree({\cal A}, w[k], w[k-1])$ \hspace*{0.5cm} \texttt{// add $w[k]$ as the son of $w[k{-}1]$ in {\cal A}}
	\ENDIF
\ENDFOR
\STATE return $\cal A$
\end{algorithmic}
\caption{Procedure $Build\_Tree(w: word): tree$ \label{algo-buildTree}}
}
\end{algorithm}

\noindent The following procedures take a rooted tree as entry (node $i$ executes this procedure): 
\begin{itemize}
\setlength{\itemsep}{-0.05cm}
\item $Tree\_To\_Word({\cal A}: tree) : word$ -- computes a word $w$ such that the first identifier of $w$ is the root of $\cal A$;
\item $My\_Sons({\cal A}: tree): identifiers\_set$ -- returns the set of $i$'s sons in $\cal A$;
\item $My\_Father({\cal A}: tree): identifier$ -- returns $i$'s father in $\cal A$.
\end{itemize}

\begin{example}
Let $G=(V,E)$ with $V=\{1,\ldots,n\}$. Consider a random-walk based token with word $w'={<}1,5,3,2,3,6,3,2,4{>}$ corresponding to a token initially generated at node 4 and arriving at node 1 after 8 movements. By algorithm \ref{algo-buildTree} the tree described in   \ref{figure-tree} \emph{(a)} is obtained.

Now, assume that the link $3-2$ disappears. At this point, the tree is not consistent with the network topology. The traversal technique using random-walk based token is used to maintain an adaptive communication structure taking into account last collected elements: if after 4 movements the word becomes $w={<}1,6,2,4,1,5,3,2,3,6,3,2,4{>}$. The tree evolves over time, cf. figure \ref{figure-tree} \emph{(b)}, taking into account the most recent data and then the tree becomes consistent with the network topology. 

The idea of the reduction  technique is to delete all useless informations in the construction of the tree. Then the word reduction of $w'$ is the following reduced word: ${<}1, 6, 2, 4, 1, 5, 3{>}$ -- we obtain the same tree. This reduction is done after each movement of the token using the procedure $Clean\_Word$, cf algorithm \ref{algo-cleanWord}. We prove in \cite{Bern06} that the size of the circulating word is bounded by $2n-1$.
\end{example}

%\begin{example}
%Let $G$ be a graph containing nodes identified from $1$ to $n$.  Let us assume a token is created by node $4$. This token circulates through a random walk in the network collecting identifiers and arrives at node $1$. Its circulation trace is collected in the word contained in  the token: $w' = {<}1,5,3,2,3,6,3,2,4{>}$.  Algorithm \ref{algo-buildTree} allows node $1$ to build from $w$ the tree in figure \ref{figure-tree} \emph{(a)}. 
%
%This allows to maintain an adaptive communication structure that take into account only the most recently collected data. Assume that the link between 2 and 3 disappears. Then, the tree is no longer consistent with the network topology. Assume that, after four moves,  the word is: $w={<}1,6,2,4,1,5,3,2,3,6,3,2,4{>}$. The tree changes, cf. figure \ref{figure-tree} \emph{(b)}, taking into account the most recent data and is consistent with the network topology again. The maintained tree is adaptive.
%
%The idea of the reduction technique is to delete all useless informations during the construction of the tree. The word reduction of $w'$ is the following reduced word: ${<}1, 6, 2, 4, 1, 5, 3{>}$ that leads to the same tree. This reduction is done after each move of the token using the procedure $Clean\_Word$, cf algorithm \ref{algo-cleanWord}. We prove that the size of the circulating node cannot be greater than $2n-1$.\textbf{où ?}
%\end{example}

\begin{figure}[H]
\begin{center}
\scalebox{0.8}{\includegraphics{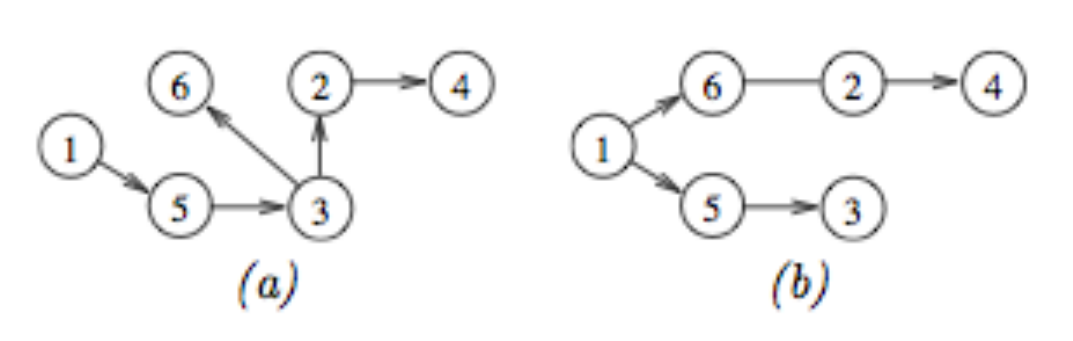}}
\caption{Construction of the tree rooted in node $1$.\label{figure-tree}}
\end{center}
\end{figure}

The procedure $Clean\_Word()$ removes all successive occurrences of $i$ in the word. It also keeps the smallest prefix of the word necessary to represent a sub-tree of the cluster, allowing to bound the size of the word. 

\begin{algorithm}[H]
{\small
\begin{algorithmic}[1]
\STATE $z \rcp 1$
\STATE $visited \rcp \{w[1]\}$
\WHILE {$z<Size(w)$}
	\IF {$w[z] = i$}
		\WHILE {$Size(w)>z+1 \wedge ((w[z+1]=w[z]) \lor (w[z+1]\notin visited \wedge w[z+1]\notin N_i))$}
		\STATE $Delete\_Element(w,z+1)$
		\ENDWHILE
	\ENDIF
	\STATE $visited \rcp visited \cup \{w[z+1]\}$
	\STATE $z \rcp z+1$
\ENDWHILE
\STATE return $w$
\end{algorithmic}
\caption{Procedure $Clean\_Word(w: word): word$ \label{algo-cleanWord}}
}
\end{algorithm}

\subsection{Clusters}

To allow scalability of random walk-based algorithms, we propose to divide the system into clusters. Each cluster is maintained by a token. The size of all clusters is greater than $m$, that is a parameter of the algorithm. Clusters are larger than $m$ nodes and as close to $m$ nodes as possible. In the following, we suppose that the system is connected and contains at least $m$ nodes.

A cluster is represented by a variable $col$ on each node. This variable is the \emph{color} (identifier) of the cluster to which the node belongs. Each cluster has to be connected.
%Following the notations given in Section \ref{variables}, we state the following definitions:
\begin{definition}[Cluster]
The cluster of color $c$, noted $V_c$, is the set of all nodes having the color $c$ (if non-empty).
$$V_c = \{i \in V, col_i=c \}$$

A non-empty set of nodes $V'$ is a cluster if there exists a color $c$ such that $V'=V_c$.
\end{definition}

\begin{definition}
We call a cluster \emph{divisible} if it can be split into two connected subgraphs of sizes greater than $m$.
\end{definition}

\begin{definition}
A spanning tree {\cal A} of a cluster is called divisible if it contains an edge that, when removed, leads to two trees of size greater than $m$.
\end{definition}

\begin{property}
A divisible cluster admits a divisible spanning tree.
\end{property}
\begin{proof}
Let $V$ a divisible cluster. Then, it admits two connected subgraphs $V_1$ and $V_2$ with sizes greater than $m$. $V_1$ (resp. $V_2$) admits a spanning tree $\mathcal A_1$ (resp. $\mathcal A_2$). $V$ being connected, consider an edge $a$ between a node in $V_1$ and a node in $V_2$. Then, the tree $\mathcal A$ obtained from $\mathcal A_1$ and $\mathcal A_2$ by adding the edge $a$ is a spanning tree of $V$ that is divisible (when removing $a$, we obtain two trees of size greater than $m$).
\end{proof}

Inthe algorithm \emph{``growing''} phase, the clusters grow from an empty set by annexing nodes. Two situations must be managed: divisible clusters, and clusters with a size strictly lower than $m$.
\begin{definition}[Stable cluster]
A cluster $V_c$ is called \emph{stable} if it is large enough:
$$ stable(V_c) = (|V_c|\geq m)$$
\end{definition}

\begin{definition}[Free Node]
A node is called free if it does not belong to any cluster:
$$free(i)= (col_i= null)$$
\end{definition}

Thus, each node in the network either belongs to a cluster or is free.

\subsection{Problem specification}

The aim of this algorithm is to build a clustering. Thus, we want that:\begin{itemize}
\item all nodes belong to a cluster;
\item clusters are connected.
\end{itemize}

Additionally, we require that all clusters have a size greater than $m$ and are not divisible.

To forbid a trivial clustering consisting in all nodes setting their color to the same color, we add an extra constraint: we require that there is no divisible cluster. No clustering algorithm working on arbitrary topologies can set both a non-trivial lower and a non-trivial upper bound on the size of the clusters. Indeed, consider a lower bound $m>1$ and an upper bound $M<n$ on a star graph on $n$ nodes. Then, there is at least two clusters (since $M<n$), and the central node is in a cluster $V$. Then, any node $i$ that is not in $V$ is in a connected cluster, that can only be $\{i\}$, which contradicts the assumption $m>1$.

Thus, we try to obtain clusters as small as possible with a size greater than $m$. Due to the distributed nature of the problem we consider, this \emph{``as small as possible''} has to be detected on a basis that is local to the cluster: no global view can be used to compute the clusters. This is why we translate \emph{``as small as possible''} in \emph{``not divisible''}

Finally, we are willing to \textbf{compute clusters of size greater than $m$ that are not divisible}.

\section{Algorithm description}

One or several nodes start the algorithm by creating a token. Each token is associated to a cluster and circulates perpetually according to a random walk scheme, building and maintaining its own cluster.
The algorithm execution goes through different phases. At the beginning, a cluster (in fact its associated token) sequentially annexes nodes, \emph{ie} when a token meets a free node, the free node joins the cluster: it is the \emph{collect mechanism}. When a token meets another cluster (maintained by another token), two cases can occur: either the token is sent back to its own cluster, or it triggers a \emph{dissolution mechanism} on its own cluster if this cluster has a size below $m$ (\emph{i.e.} a non-stable cluster). The goal of this mechanism is to delete clusters that cannot reach a size of $m$ by making all their nodes free. The third mechanism is the \emph{division mechanism}: if a cluster grows and becomes divisible, then the cluster is divided into two smaller stable clusters.

The algorithm uses five types of messages: \emph{Token, Dissolution, FeedbackDiss, Division} and \emph{FeedbackDiv}. \emph{Token} messages are the main messages of the algorithm. They represent a cluster and collect node identifiers  during their circulation (algorithms \ref{algo-create} and \ref{algo-token}). The four other types of messages are control messages circulating in a cluster built by a token. The dissolution and division mechanisms are based on a classical propagation of information with feedback \cite{Sega83, Tel94, BDPV07} over the cluster that has to be dissolved or divided. The dissolution mechanism uses \emph{Dissolution} and \emph{FeedbackDiss} messages to make all nodes in a cluster free  (algorithms \ref{algo-diss} and \ref{algo-feedbackDiss}). The division mechanism uses \emph{Division} and \emph{FeedbackDiv} messages to divide one large cluster into two smaller clusters (of size still greater than $m$) (algorithms \ref{algo-div} and \ref{algo-feedbackDiv}). 

\subsection{Nodes and tokens variables}\label{variables}

A cluster is identified by a \emph{color}: each node having this color belongs to this cluster. The color of a cluster is the identifier of the node that has created this cluster. A cluster is created when a free node awakens and creates a token (cf. Algorithm \ref{algo-create}).

\noindent Each token $T$ contains:
\begin{itemize}
\setlength{\itemsep}{0cm}
\item $col_T$: the color (identifier) of the cluster it represents and
\item $w_T$: the set of nodes that belongs to this cluster (this is the circulating word).
\end{itemize}
The token gathers identities of the visited nodes as explained in the section above.

A node $i$ saves the identifier of its cluster, and the identifiers of the nodes belonging to this cluster. 
The local variables of a node having the identifier $i$ are:
\begin{itemize}
\setlength{\itemsep}{0cm}
\item $col_i$: color of the cluster $i$ belongs to;
\item $w_i$: circulating word of the last token with the same color as $i$ that has visited $i$;
\item $nbFeedbackDiv_i$: number of $FeedbackDiv$ messages received during the division phase;
\item $nbFeedbackDiss_i$: number of $FeedbackDiss$ messages received during the dissolution phase.
\end{itemize}
The last two variables are necessary to the {Propagation of Information with Feedback} algorithms used in the dissolution and division mechanisms.

The definition domain of $col_i$ is: $\{$node identifiers in the network$\}\cup \{null, -1\}$. We assume there is no node in the network having the identifier -1.
A free node is such that $col_i=null$. A \emph{locked} node is such that $col_i=-1$. 
The difference between a free node and a locked node is that a free node does not belong to any cluster and can join one, while a locked node belongs to a ``false'' cluster, which forbids the node to join a real one. The locked state is used in the dissolution mechanism and it is such that no node will remain locked forever.

The variable $w_i$ is used to break the symmetry in the dissolution mechanism (cf. section \ref{section-collect}) and avoid reciprocal destructions. It also allows node $i$ to know all identifiers of its cluster.

\subsection{Algorithm}

In the following, all procedures are executed by node $i$.

Initially, all nodes are free.

\begin{algorithm}[H]
{\small
\begin{algorithmic}[1]
	\STATE $col_i \rcp null $
	\STATE $w_i \rcp \varepsilon$
	\STATE $nbFeedbackDiv_i \rcp 0$
	\STATE $nbFeedbackDss_i \rcp 0$
\end{algorithmic}
\caption{On initialization
\label{algo-init}}
}
\end{algorithm}

\subsubsection{Collect Mechanism} \label{section-collect}

When awakening, a free node creates a token with probability $1/2$; otherwise it goes back to sleep for a random time. The color of the created token is the identifier of the node that creates the token.

\begin{algorithm}[H]
{\small
\begin{algorithmic}[1]
	\IF{$col_i = null$} 
		\STATE Toss a coin
		\IF{tail}
			\STATE $col_i \rcp i$ \label{create-deb}
			\STATE $w_i \rcp {<}i{>}$
			\STATE Random choice of $j \in N_i$, Send $Token(col_i, w_i)$ to $j$ \label{create-fin}
		\ENDIF
	\ENDIF
\end{algorithmic}
\caption{On node $i$ awakening
\label{algo-create}}
}
\end{algorithm}

To describe the various cases of the algorithm \ref{algo-token}, we assume that node $i$ receives a $red$ token: 

\paragraph{If $i$ is $red$ or free (Algorithm \ref{algo-token} lines \ref{token-case1} to \ref{token-joinFin}):}
First $i$ adds its own identifier at the beginning of the circulating word and then cleans this word (Algorithm \ref{algo-token} lines \ref{token-add} and \ref{token-check}). This ensures that the circulating word always represents a spanning tree of the whole cluster.

Second, if the tree represented by the word in the token is divisible, then $i$ launches a division: a spanning tree is built from the circulating word and the division is done over this spanning tree (algorithm \ref{algo-token} lines \ref{token-divDeb} to \ref{token-divFin}). Otherwise $i$ joins the cluster and sends the token to one of its neighbor chosen  uniformly at random (algorithm \ref{algo-token} lines \ref{token-joinDeb} to \ref{token-joinFin}). 

In the following, more details about the division are given (lines \ref{token-divDeb} to \ref{token-divFin}). The spanning tree is partitioned into two sub-trees using the following procedure $Divide({\cal A}~tree): (w_1: word, w_2: word)$:
\begin{enumerate}
\item ${\cal A}_1$ and ${\cal A}_2$ are two subtrees of $\cal A$;
\item the union of ${\cal A}_1$ and ${\cal A}_2$ is $\cal A$ and the intersection of ${\cal A}_1$ and ${\cal A}_2$ is empty;
\item the root of ${\cal A}_1$ is $i$;
\item $w_1$ is a word representing ${\cal A}_1$;
\item $w_2$ is a word representing ${\cal A}_2$;
\item $Nb\_Identities(w_1) \geq m$ and $Nb\_Identities(w_2) \geq m$.
\end{enumerate}

Such an algorithm can, for example, associate to each node in $\mathcal A$ the size of the subtree rooted in this node. If for a node $i$, this size $k$ is such that $k\geq m$ and $Nb\_identities(\mathcal A)-k\geq m$, then it returns the word computed from the subtree rooted in $i$ and the word computed from the subtree obtained by removing $i$ and its descendants from $\mathcal A$. If no node verifies this, then the tree is not divisible.

\centerline{\scalebox{.5}{\includegraphics{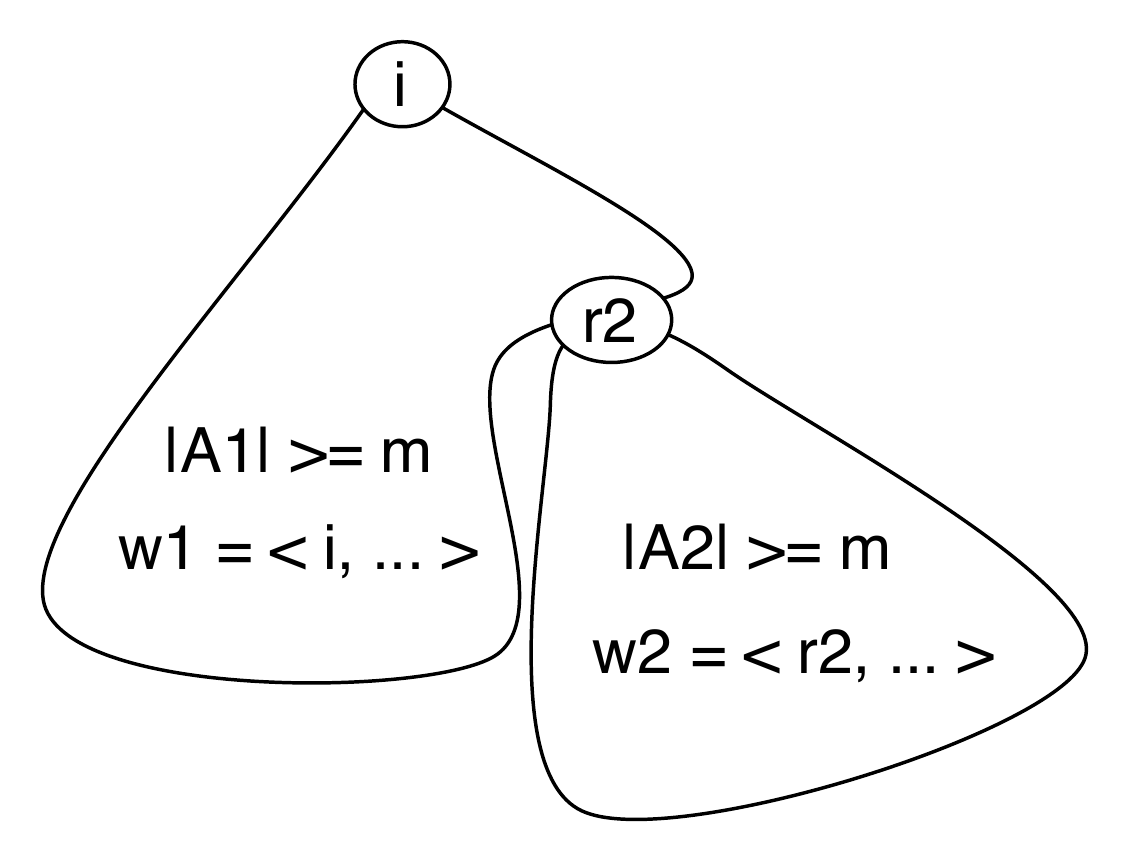}}}
The node $i$ uses the procedure $Is\_Divisible$ in order to know if the spanning tree is divisible:\\
\centerline{Procedure $Is\_Divisible({\cal A}~tree) : boolean$ -- returns $true$ iff $Divide({\cal A})$ is possible}

If the node $i$ launches a division, then $i$ joins the cluster represented by the word $w_1$. The color of this new cluster is $i$. Then $i$ initiates a propagation of $Division$ messages in the spanning tree of the whole cluster (\emph{i.e.} the tree {\cal A}).

\paragraph{If $i$ is locked (Algorithm \ref{algo-token} lines \ref{token-falseDeb} and \ref{token-falseFin}):}
$i$ sends the token back to its sender.

\paragraph{If $i$ is $blue$ (Algorithm \ref{algo-token} lines \ref{token-case2Deb} to \ref{token-case2Fin}):}
If the size of the cluster represented by $T$, noted $\phi_T$, is too small ($|w_T|<m$), then the cluster is dissolved in order for the cluster to which $i$ belongs to grow. This dissolution is achieved by $i$ launching a dissolution mechanism. However, if the size of $i$'s cluster, noted $\phi_i$, is also too small, then we have to avoid the case when $\phi_i$ is destroyed by a node in $\phi_T$ and $\phi_T$ is destroyed by a node in $\phi_i$. Thus $i$ can launch a dissolution mechanism over $\phi_T$ only if the size of $w_i$ is large enough ($w_i\geq m$) or if $i$'s color is greater than the one of $T$ (lines \ref{token-dissDeb} to \ref{token-dissFin}).

Note that $w_i$ does not contain all node identifiers in $\phi_i$, but a subset of it. Indeed when the token representing $\phi_i$, noted $t$, arrived in $i$ for the last time, $i$ saved $w_t$ in $w_i$. Then all identifiers in $w_i$ are identifiers of $\phi_i$. However, $t$ kept circulating after arriving in $i$ and then some free nodes kept joining $\phi_i$. Thus, some identifiers in $\phi_i$ may not belong to $w_i$.

If $i$ does not launch a dissolution mechanism, then it sends the token back to its sender (line \ref{token-sendBack}).

\begin{algorithm}[H]
{\small
\begin{algorithmic}[1]
	\IF{$(col_i = null \lor col_i = col_T)$ \label{token-case1}}
		\vspace*{5pt}
	\STATE \texttt{//------ Case 1: (i is free) or (i is red and receives a red token)}\\[5pt]
		\STATE $w_T \rcp Add\_Begin(i,w_T)$ \label{token-add}
		\STATE $w_T \rcp Clean\_Word(w_T)$ \label{token-check}  \\[5pt]
		\IF{$Is\_Divisible(w_T)$  \label{token-divDeb}} 
			\vspace*{5pt}
			\STATE \texttt{// If the cluster is large enough, we launch a division}\\
			\STATE ${\cal A} \rcp Build\_Tree(w_T)$
			\STATE $(w_1, w_2) \rcp Divide({\cal A})$ 
				\texttt{~~~~~~~ // $w_1={<}i, \ldots{>}$ and $w_2={<}r2, \ldots{>}$}
			\STATE $col_i \rcp i$
			\STATE $w_i \rcp w_1$
			\STATE $\forall j \in My\_Sons({\cal A}):$ Send 
					$Division({\cal A}, w_1, w_2)$ to $j$ \label{token-divFin} \\[5pt]
		\ELSE \label{token-joinDeb}
		\vspace*{5pt}
		\STATE \texttt{// Otherwise i joins the cluster and forwards the token}\\
			\STATE $col_i \rcp col_T$
			\STATE $w_i \rcp w_T$
			\STATE Random choice of $j \in N_i$, Send $Token (col_T, w_T)$ to $j$ \\[5pt]
		\ENDIF \label{token-joinFin}
		\vspace*{5pt}
	\ELSIF{$col_i=-1$ \label{token-falseDeb}}
			\STATE  Send $Token (col_T, w_T)$ to $w_T[1]$ \label{token-falseFin}\\[5pt]
		\STATE \texttt{//------ Case 2: i is blue and receives a red token} \label{token-case2Deb}\\[5pt]
	\ELSE
		\vspace*{5pt}
		\STATE \texttt{// If the red cluster is too small and under some asymmetric assumptions, 
			$i$ can dissolve it\\
			// Otherwise $i$ sends back the token to its sender}\\
		\IF{$Nb\_Identities(w_T)<m \land (Nb\_Identities(w_i)\geq m \lor col_i > col_T)$ \label{token-dissDeb}}
			\STATE ${\cal A} \rcp Build\_Tree(w_T)$
			\STATE  Send $Dissolution({\cal A})$ to $w_T[1]$ \label{token-dissFin}
		\ELSE 
			\STATE  Send $Token (col_T, w_T)$ to $w_T[1]$ \label{token-sendBack}
		\ENDIF 		
	\ENDIF \label{token-case2Fin}
\end{algorithmic}
\caption{On the reception of $Token(col_T, w_T)$
\label{algo-token}}
}
\end{algorithm}

\subsubsection{Dissolution mechanism}

The dissolution mechanism is used to totally delete a cluster with a size smaller than $m$. This mechanism is a classical propagation of information with feedback. A $Dissolution$ message is propagated through a spanning tree of the cluster to dissolve, and then during the feedback phase, a $FeedbackDiss$ message is propagated in the tree. Each of these two kinds of messages has one variable, $\cal A$ that is a rooted spanning tree of the cluster that has to be dissolved. Both propagation and feedback waves are propagated on the tree $\cal A$.

During the propagation phase, a node receiving a $Dissolution$ message leaves its cluster and becomes locked. The propagation phase is used to completely delete the cluster. Then during the feedback phase, a node receiving a $FeedbackDiss$ message becomes a free node.  At the end of a dissolution, all nodes that used to be in the cluster are free, and the cluster does not exist anymore.

\begin{algorithm}[H]
{\small
\begin{algorithmic}[1]
	\STATE $col_i \rcp -1$
	\STATE $w_i \rcp \varepsilon $
	\vspace*{5pt}
	\IF{$|My\_Sons({\cal A})| > 0$}
		\STATE $\forall j \in My\_Sons({\cal A}) :$ Send $Dissolution({\cal A})$ to $j$
		\STATE $nbFeedbackDiss_i \rcp 0$
	\ELSE
		\STATE Send $FeedbackDiss({\cal A})$ to $My\_Father({\cal A})$
		\STATE $col_i \rcp null$
	\ENDIF
\end{algorithmic}
\caption{On the reception of $Dissolution({\cal A})$
\label{algo-diss}}
}
\end{algorithm}

\begin{algorithm}[H]
{\small
\begin{algorithmic}[1]
	\STATE $nbFeedbackDiss_i ++$\\
	\texttt{// if $i$ receives the last feedback it was waiting for}
	\IF{$nbFeedbackDiss_i = |My\_Sons({\cal A})|$}
		\STATE send $FeedbackDiv({\cal A})$ to $My\_Father({\cal A})$
		\STATE $col_i \rcp null$
	\ENDIF
\end{algorithmic}
\caption{On the reception of $FeedbackDiss({\cal A})$
\label{algo-feedbackDiss}}
}
\end{algorithm}

\subsubsection{Division mechanism}

The division mechanism is used to divide a cluster $\phi$ into two smaller clusters $\phi_1$ and $\phi_2$ with a size greater than $m$. This mechanism is a Propagation of Information with Feedback (PIF). A $Division$ message is propagated in a spanning tree of the cluster to divide, and then during the feedback phase, a $FeedbackDiss$ message is propagated through the tree. Each of these two kind of messages has three variable:
\begin{itemize}
\item $\cal A$: a rooted spanning tree of the cluster that has to be dissolved;
\item $w_1$: a rooted spanning tree containing all node identifiers of the first sub-cluster $\phi_1$;
\item $w_2$: a rooted spanning tree containing all node identifiers of the first sub-cluster $\phi_2$.
\end{itemize}

Subtrees $w_1$ and $w_2$ are a partition of the tree $\cal A$. Moreover, the first identifier of $w_1$ (resp. $w_2$) is the root of the sub-tree built by $w_1$ (resp $w_2$). The color of the cluster $\phi_1$ (resp $\phi_2$) is the identifier of the root of the trees built by $w_1$ (resp. $w_2$). The PIF is propagated among the tree $\cal A$.

During the propagation phase, a node receiving a $Division$ message checks in which subtree, $w_1$ or $w_2$, it belongs to, and joins this cluster accordingly (algorithm \ref{algo-div}, lines \ref{div-joinDeb} to \ref{div-joinFin}). Then the node executes the PIF algorithm (algorithm \ref{algo-div}, lines \ref{div-pifDeb} to \ref{div-pifFin}). When the root of $w_1$ (resp. $w_2$) receives the last $FeedbackDiv$ message it is waiting for, the node creates a new token initialized with the word $w_1$ (resp. $w_2$) (algorithm \ref{algo-feedbackDiv}).

The division mechanism is based on a PIF rather than a diffusion, because when a division is launched, the token of the cluster disappears. The new tokens are only created during the feedback, so that when they are created, they correspond to the existing clusters.

\begin{algorithm}[H]
{\small
\begin{algorithmic}[1]
	\IF{$i \in w_1$ \label{div-joinDeb}}
		\STATE $w_i\rcp w_1$
		\STATE $col_i\rcp w_1[1]$
	\ELSE
		\STATE $w_i\rcp w_2$
		\STATE $col_i\rcp w_2[1]$
	\ENDIF \label{div-joinFin}
	\vspace*{5pt}
	\IF{$|My\_Sons({\cal A})| > 0$ \label{div-pifDeb}}
		\STATE $\forall j\in My\_Sons({\cal A})$ : Send $Division({\cal A},w_1,w_2)$ to $j$
		\STATE $nbFeedbackDiv_i \rcp 0$
	\ELSE
		\STATE Send $FeedbackDiv({\cal A},w_1,w_2)$ to $My\_Father({\cal A})$
	\ENDIF \label{div-pifFin}
\end{algorithmic}
\caption{On the reception of $Division({\cal A},w_1,w_2)$
\label{algo-div}}
}
\end{algorithm}

\begin{algorithm}[H]
{\small
\begin{algorithmic}[1]
	\STATE $nbFeedbackDiv_i ++$\\
	\texttt{// if $i$ receives the last feedback it was waiting for}
	\IF{$(i\in w_1 \land nbFeedbackDiv_i = |My\_Sons(w_1)|) \lor (i\in w_2 \land nbFeedbackDiv_i = |My\_Sons(w_2)|)$}
		\IF{$i=w_1[1]$}
			\STATE \texttt{// if $i$ is the root of the first tree}
			\STATE Random choice of $j \in N_i$, send $Token(i, w_1)$ to $j$
		\ELSIF{$i = w_2[1]$}
			\STATE \texttt{// if $i$ is the root of the second tree}
			\STATE Random choice of $j \in N_i$, send $Token(i, w_2)$ to $j$
		\ELSE
			\STATE send $FeedbackDiv({\cal A}, w_1, w_2)$ to $My\_Father({\cal A})$
		\ENDIF
	\ENDIF
\end{algorithmic}
\caption{On the reception of $FeedbackDiv({\cal A}, w_1, w_2)$
\label{algo-feedbackDiv}}
}
\end{algorithm}

\section{Proof of correctness}

In this section, we prove that starting from an initial configuration, the clustering eventually meets the problem specification (see subsection 2.5).

\subsection{Preliminaries}

Assuming $G=(V,E)$ is the communication graph and $c$ a color, we have the following definitions.

\begin{definition}[Graph induced by a cluster]
Let $V_c$ be a cluster. We note $G_c$ the graph induced by all the nodes in $V_c$.
$$G_c = (V_c, E_c) \mbox{ with: } E_c = E \cap (V_c\times V_c)$$
\end{definition}

\begin{definition}[Cluster neighbors]
Let $V_c$ be a cluster. The neighborhood of $V_c$, noted $N(V_c)$, is:
$$N(V_c) = \{i\in V, \exists j\in V_c, j\in N_i\}\cup V_c$$
\end{definition}

Note that $V_c\subset N(V_c)$.

\begin{definition} A configuration $\gamma$ is called legitimate iff:
\begin{itemize}
\item each node belongs to exactly one cluster;
\item each cluster is connected;
\item each cluster has a size greater than $m$;
\item no cluster is divisible;
\item there is exactly one token of color $c$ in $N(V_{c})$;
\item there are only $Token$ messages circulating in the network.
\end{itemize}
\end{definition}

\begin{property}
A legitimate configuration respects the specification of the problem:
\begin{enumerate}
\item all nodes belong to a cluster;
\item all clusters are connected;
\item each cluster has a size greater than $m$;
\item no cluster is divisible.
\end{enumerate}
\end{property}

\begin{definition} A configuration is called initial iff:
\begin{itemize}
\item each node is free;
\item there is no message circulating in the network.
\end{itemize}
\end{definition}

\subsection{Correctness proofs}

We show in this section that from the initial configuration the system reaches a legitimate configuration.

At the initialization, there is no message in the network. Then the only rule a node can execute is the algorithm \ref{algo-create}. Only free nodes can execute this rule, and its execution makes the node create a token with probability $\frac12$. Thus a node $i$ will eventually execute this rule and will then create a token $T = Token(i,<i>)$. 

At this point a cluster $V_i$ is created, and node $i$ takes the color $i$. $T$ begins to circulate in the network. We say that $V_i$ is in its \emph{growing} phase. Each time $T$ is received by a node, this node (i) enters in the cluster writing $i$ in its $col$ variable, adds its id at the beginning of the word $w_T$ and makes $T$ circulate, or (ii) sends back $T$ to the sender, or (iii) makes the cluster change its phase, going from phase \emph{growing} to the phase \emph{division} ($T$ disappears and the node sends back a division message) or \emph{dissolution} ($T$ disappears and the node sends back a division message). Thus, during the \emph{growing} phase, the cluster $V_i$ is connected, there is only one token of color $i$ in the network and $w_T$ is a spanning tree of $V_i$.

\begin{lemma}
Except when being dissolved or divided, clusters are connected.
\end{lemma}
\begin{proof}
Consider a cluster $V_c$. $V_c$ evolves by being created, recruiting new nodes, being dissolved, or being split.

When created, $V_c=\{c\}$ (no other node can have color $c$, since this would mean that $c$ had already created a token, in which case it can only get free again after the dissolution of the cluster).  $V_c$ is connected.

A new node is recruited to $V_c$ when it receives a token with color $c$, sent by a node that is already in $V_c$: indeed, a token with color $c$ can only be sent by a node that is already in $V_c$ (algorithm 4 line 6, algorithm 5 lines 16 and 19, algorithm 9 lines 5 and 8), or by a node neighboring $V_c$ that sends the token back to its sender (already in $V_c$). Thus, there is a link between a newly recruited node and a node that is already in $V_c$: if $V_c$ is connected, it remains so.

When $V_c$ is dissolved, a PIF is launched on one of its spanning trees, at the end of which all nodes have left $V_c$: $V_c$ is then the empty graph, which is connected.

When $V_c$ is divided, it is divided only when divisible, \emph{ie} when it can be divided into two connected clusters of size greater than $m$. Then all nodes leave $V_c$ and join their new clusters, that are connected (one of which may still have color $c$).
\end{proof}

\begin{lemma}
At least one cluster is created \emph{whp}.
\end{lemma}
\begin{proof}
If there is no cluster in the network, on a free node awakening, a token is created infinitely often with probability $1/2$: eventually, \emph{whp}, a token is created. Then this token starts building a cluster (\emph{cf.} Algorithm \ref{algo-create}). 
\end{proof}

\begin{lemma}\label{unique}
Eventually, a token of color $c$ exists if and only if $V_c\neq\emptyset$. In this case, this token is unique and is in $N({V_c})$.
%There is at most one token of color $c$, and, if it exists, this token is in $N({V_c})$
\end{lemma}
\begin{proof}
The only token creation is on a free node executing algorithm \ref{algo-create}. The color of this token is the id of the free node that creates it. Since, after creating a token, a node has a color, it can no longer be free. The only cases when a node becomes free again are in algorithms 6 and 7 (in algorithm 8, lines 3 and 6, $w_1[1]$ and $w_2[1]$ cannot be null, since the words computed by procedure $Division$ are non-empty). To execute algorithms 6 and 7, a node has to have received a message $Dissolution$ or a message $FeedbackDiss$ respectively. A $Dissolution$ message can only be sent by its father if it has itself received this message (algorithm 6) or by a node in a stable cluster or with a higher color, in which case the token of its color has been removed (algorithm 5). Messages $FeedbackDiss$ can only be received from a son to which the node has already sent a $Dissolution$ message. Anyway, the dissolution procedure is always triggered on a cluster of color $i$ after a token of color $i$ has been removed, and it is the only phase when a node can become free. Thus, after a node $i$ has created a $Token$ with color $i$, it can create a new $Token$ only if the previous one has been removed. Since it is the only node that can create a token with color $i$, there is at most one token with color $i$ at a time in the network.

Thus, if a token of color $c$ exists, the node of id $c$ is in $V_c$. And if $V_c\neq\emptyset$, then $V_c$ contains $c$ and there is a token of color $c$.

If a node outside $V_c$ receives a $Token$ with color $c$, it either joins $V_c$ (if it is free; algorithm 5, line 9), or triggers the removal of this token and the dissolution of $V_c$ (algorithm 5, line 25), or sends it back to its sender (algorithm 5, line 27). Since the token starts in $V_c$ (algorithm 4), the token can only be in $V_c$, or on a neighboring node.
\end{proof}

\begin{lemma}\label{clusterstable}
Eventually, one cluster (at least) is stable \emph{whp}.
\end{lemma}
\begin{proof}
In a given execution of this algorithm, there is a finite number of clusters colors used (at most, as many colors as nodes in the system, since the color of a cluster is the ID of the node that initiated it).

Consider $c$ the highest color appearing in an execution. Suppose, for the sake of contradiction, that no cluster ever becomes stable.

Suppose that, in the execution, a cluster with color $c$ is dissolved: then, the dissolution process can only have been initiated by a clustered node neighboring $V_c$, on the reception of the token with color $i$ (algorithm 5, line 25). It launches the dissolution only if it is in a cluster with a higher color than $c$, which is discarded by the definition of $c$, or if it is in a stable cluster (line 23). Thus, a cluster of color $c$ cannot be dissolved.

Now, suppose that no cluster with color $c$ is dissolved. Then, the token with color $c$ is never removed (see proof of the previous lemma). Thus, the token follows an infinite path. Each time the token hits a free node, this node is added to the cluster: $|V_c|$ is incremented, and the node is no longer free, and will remain in $V_c$ until a division occurs (by assumption, $V_c$ is not dissolved in the considered execution), which is possible only if the cluster is stable. Now, \emph{whp}, any unstable neighboring cluster is dissolved, when its token reaches a node in $V_c$, for example (which will occur \emph{whp} according to the hitting property of random walks). At this moment, there are free nodes in $N(V_c)$. The token with color $c$ following a random walk on nodes that are either free or in $V_c$, it will reach such a node \emph{whp} and recruit it to $V_c$. Since the moves of the different tokens are independent, \emph{whp}, the token will be able to recruit a new node. Once $m$ nodes are recruited, $V_c$ is stable, which contradicts the assumption.

Thus, \emph{whp}, at least one cluster is eventually stable.
\end{proof}

\begin{lemma}\label{stablenode}
A node belonging to a stable cluster remains in a stable cluster.
\end{lemma}
\begin{proof}
Consider a node $i$ belonging to a stable cluster of color $c$. Since its color is not $null$, $i$ can only change its color in algorithm 5 line 9, algorithms 6 and 7, or algorithm 8, line 3 and 6. Algorithms 6 and 7 correspond to a dissolution phase of $V_c$, which is impossible since $V_c$ is stable. Algorithm 5 line 9, and algorithm 8, are part of the division process, that divides a stable cluster into two stable clusters (algorithm 5, line 5: the process is launched only when the tree is splittable into two subtrees with sizes greater than $m$). Thus, the cluster to which $i$ belongs after the division is still stable.
\end{proof}

\begin{lemma}
A divisible cluster is eventually divided \emph{whp}.
\end{lemma}
\begin{proof}
Consider a divisible cluster $V_c$, and suppose it is never divided. Since it is stable, its size can only increase. Suppose it has reached its maximal size. Then, its token browses $N(V_c)$, and if we put apart the steps when the token goes out of $V_c$ and is sent back, it follows a random walk on $V_c$. The spanning tree of $V_c$ is computed according to the last minimal terminal covering path of the token: according to \cite{Aldo90}, it contains a random spanning tree of $V_c$. Thus, any spanning tree of $V_c$ is computed \emph{whp}. Now, at least one spanning tree is divisible (property 1).

Thus, \emph{whp}, the tree computed from the circulating word is eventually divisible, and the cluster is divided according to algorithm 5 line 7 to 11.
\end{proof}

\begin{lemma}
The number of nodes in stable clusters grows to $n$ \emph{whp}.
\end{lemma}
\begin{proof}
According to lemma \ref{clusterstable}, eventually and \emph{whp}, there exists a stable cluster. Suppose that some nodes never belong to a stable cluster, and consider a node $i$ that never belongs to a stable cluster, and neighboring a node $j$ in a stable cluster. Such a node exists since the network is connected. $j$ remains in a stable cluster according to lemma \ref{stablenode}. At some point, and by the assumption above, the algorithm stops recruiting new nodes to stable clusters. Then, at some point, stable clusters do not evolve any longer:
\begin{itemize}
\item they cannot recruit new nodes;
\item they cannot be dissolved since they are stable;
\item they can no longer be split since they do not grow any more (once all splittable clusters are split, which eventually happens --- see demonstration of the previous lemma --- they can no longer be split).
\end{itemize}

Note $V_c$ the cluster of $j$.

There is a random walk browsing $N(V_c)$, and, \emph{whp}, it will hit $i$ infinitely often. Since $i$ is not in a stable cluster, \emph{whp} it is infinitely often free (indeed, if it is in a cluster, the token in this cluster must hit $V_c$ infinitely often, and the cluster is dissolved infinitely often, which makes $i$ free). Thus, \emph{whp} $i$ is eventually hit by the token of $V_c$ while being free (by independency of the moves of the different tokens), and is then recruited to the stable cluster $V_c$, which contradicts the assumption.

Thus, all nodes eventually belong to a stable cluster.
\end{proof}

\begin{lemma}\label{st}
All clusters eventually have a size greater than $m$, and are not divisible \emph{whp}.
\end{lemma}
\begin{proof}
Once all node are in stable clusters, all clusters have a size greater than $m$ by definition. Then, according to lemma 6, all divisible clusters are divided \emph{whp}: eventually none of them is divisible.
\end{proof}

\begin{theorem}
The algorithm converges to a legitimate state \emph{whp}.% clustering with all clusters greater than $m$. If the assumption ??? holds, then all clusters have a size lesser than $M$.
\end{theorem}
\begin{proof}
This comes from lemmas 1, 7 and 8.
\end{proof}

\begin{corollary}
Consider an integer $M\geq 2m$.
If $G$ is such that all its connected subgraphs of size greater than $M$ can be partitioned into two connected subgraphs of size greater than $m$, then the algorithm converges \emph{whp} to a clustering such that all clusters have a size between $m$ and $M-1$.
\end{corollary}
\begin{proof}
Indeed, in such a graph, all clusters of size greater than $M$ are divisible.
\end{proof}

For example, on complete graphs or on rings, the algorithm leads to a clustering with clusters of sizes between $m$ and $2m-1$.

\section{Adaptive algorithm}

%Under an arbitrary mobility assumption, there exists no efficient reconfigurable clustering algorithm. Indeed, if the topology changes occur more frequently than the algorithm adapts the different clusters of the network, clusters will never been reconfigured. Thus we are lead to work with reasonnable mobility assumption.

This algorithm can be made adaptive to mobility. In this section, we call correct configuration any configuration that could have arisen from the execution of the algorithm on a static network corresponding to the current state of the network. More precisely:
\begin{definition}
A configuration is called correct if:\begin{itemize}
\item each cluster is connected;
\item there is exactly one token of color $c$ in $N(V_{c})$;
\item if node $i$'s color is $c\neq null$, $i$ is in the word of the token of color $c$, and only in this word; if $i$ is free, then it does not occur in any circulating word;
\item if $i$ is the son of $j$ in the tree of some token, then there is a link $(i, j)\in E$;
\item no message $Delete$ circulates.
\end{itemize}
\end{definition}

A correct configuration is not always legitimate, but, without any further topological modifications, the previous section proves that from any correct configuration, the system eventually reaches a legitimate configuration.

Mobility can manifest in the following ways:\begin{itemize}
\item link connection: the configuration remains correct, and nothing has to be done;
\item node connection: since the node that connects is initially free, it does not appear in any circulating word and the configuration remains correct;
\item node disconnection: this case is dealt like the disconnection of all adjacent links;
\item link $(i, j)$ disconnection: \begin{itemize}
	\item if this link is between two clusters ($col_i\neq col_j$, $col_i\neq null$, $col_j\neq null$): the link $(i,j)$ can appear neither on $w_{T_{col_i}}$ nor on on $w_{T_{col_j}}$; we use an acknowledgement message to detect the case when a token goes (for instance) from $i$ to $j$, and a disconnection occurs before it is sent back to $i$, leading to no token remaining in the cluster of $i$; in this case, $j$ deletes the token, and $i$ creates a token with the same content as the one of the deleted token;
%	. But as the token of color $col_i$ moves among $N(V_{col_i})$ it is possible that token $T_{col_i}$ moves on node$j$ just before the link $(i,j)$ disappear and stays trapped on cluster of color $col_j$. The configuration is no longer correct since there is no token of color $col_i$ in the cluster of color $col_i$. At this step, we need a minimalism assumption on the mobility: if a token is sent through a channel, the node receiving it could answer to the sender. Thus, the token can be returned to its sender, and the configuration remains correct;
	\item if one of this link extremities (or both) is a free node ($col_i=null\vee col_j=null$): since $j$ is free, the link $(i,j)$ could not appear on any circulating word: the configuration is still correct;
	\item if this link is between two nodes in the same cluster ($col_i=col_j\neq null$):\begin{itemize}
		\item if they are not linked by any father-son relationship: the link $(i,j)$ does not appear in the token of the cluster of color $col_i$, and the configuration remains correct;
		\item if $i$ is the father of $j$: the connectivity of the tree contained in the token of color $col_i$ is broken; the configuration is no longer correct, and the different mechanisms we have setup so far cannot work correctly: this is the reason why we introduce the mechanisms presented in this section, to deal with this case.
		\end{itemize}
	\end{itemize}
\end{itemize}

%On a node or a link connecting, the configuration remains correct, and no further work is needed to tolerate a new connection. What we have to manage is links and nodes disconnections.

%When a node $i$ detects the disappearance of a link $(i, j)$, then if $i$ and $j$ are in different clusters, nothing happens. 
%If $i$ and $j$ are in the same cluster, but have no father-son relation, nothing happens neither. 
If $j$ is its father, $i$ deletes the subtree rooted in itself: in other words, it propagates a wave on this subtree, that makes all its descendants free. Once $i$ has set itself and all its descendant free, the token is no longer correct. When it visits $j$ for the next time, $j$ corrects it by removing the subtree rooted in any son it is supposed to have, but to which it is actually not connected. Thus, the system reaches a correct configuration.

To implement this, $i$ must always know which node is its father in the token with the same color as it. If $j$ is the father of $i$ in this tree, it means that the last time $i$ owned the token, it transferred it to $j$. Thus, $i$ can remember which node is its father.

The algorithms below describe this mobility-adaptive distributed algorithm. The disconnection of a node is dealt as the disconnection of all its adjacent links.

The initialization phase is left unmodified, except for the initialization of the variable $father$.

\begin{algorithm}[H]
{\small
\begin{algorithmic}[1]
	\STATE $col_i \rcp null $
	\STATE $w_i \rcp \varepsilon$
	\STATE $nbFeedbackDiv_i \rcp 0$
	\STATE $nbFeedbackDss_i \rcp 0$
	\STATE $father_i\rcp null$
	\STATE $version \rcp 0$
\end{algorithmic}
\caption{On initialization
\label{algo-initM}}
}
\end{algorithm}

\begin{algorithm}[H]
{\small
\begin{algorithmic}[1]
	\IF{$col_i = null$} 
		\STATE Toss a coin
		\IF{tail}
			\STATE $col_i \rcp (i, version++)$ \label{create-deb}
			\STATE $w_i \rcp {<}i{>}$
			\STATE Random choice of $j \in N_i$, Send $Token(col_i, w_i)$ to $j$ \label{create-fin}
			\STATE $father_i\rcp j$
		\ENDIF
	\ENDIF
\end{algorithmic}
\caption{On node $i$ awakening
\label{algo-createM}}
}
\end{algorithm}

When receiving a token, a node first checks whether it is correct regarding its own neighborhood. If it detects an inconsistency between the tree borne by the token and its neighborhood, it modifies the tree so that it is in accordance with the (new) topology of the system.

\begin{algorithm}[H]
{\small
\begin{algorithmic}[1]
	\STATE ${\cal A} \rcp Build\_Tree(w_T)$
	\IF{$\mathcal A$ contains an edge $(j, i) \wedge j\notin N(i)$}
		\STATE \texttt{// If $i$ is supposed to be $j$'s father, but is not actually connected to it}\\
		\STATE remove the subtree of $\mathcal A$ rooted in $j$
		\STATE $w_T\rcp Tree\_To\_Word(\mathcal A)$
%		\IF{$col_i\notin w_T$}
%			\STATE $col_i\rcp null$
%			\STATE $father_i\rcp null$
%			\STATE $\forall j\in My\_Sons(\mathcal A)\cap N_i, $ send $Dissolution(\mathcal A)$ to $j$
%		\ENDIF
	\ENDIF
	\IF{$(col_i = null \lor col_i = col_T)$ \label{token-case1}}
		\vspace*{5pt}
	\STATE \texttt{//------ Case 1: (i is free) or (i is red and receives a red token)}\\[5pt]
		\STATE $w_T \rcp Add\_Begin(i,w_T)$ \label{token-add}
		\STATE $w_T \rcp Clean\_Word(w_T)$ \label{token-check}  \\[5pt]
		\IF{$Is\_Divisible(w_T)$  \label{token-divDeb}} 
			\vspace*{5pt}
			\STATE \texttt{// If the cluster is large enough, we launch a division}\\
			\STATE $(w_1, w_2) \rcp Divide({\cal A})$ 
				\texttt{~~~~~~~ // $w_1={<}i, \ldots{>}$ and $w_2={<}r2, \ldots{>}$}
			\STATE $col_i \rcp i$
			\STATE $w_i \rcp w_1$
			\STATE $\forall j \in My\_Sons({\cal A}):$ Send 
					$Division({\cal A}, w_1, w_2)$ to $j$ \label{token-divFin} \\[5pt]
		\ELSE \label{token-joinDeb}
		\vspace*{5pt}
		\STATE \texttt{// Otherwise i joins the cluster and forwards the token}\\
			\STATE $col_i \rcp col_T$
			\STATE $w_i \rcp w_T$
			\STATE Random choice of $j \in N_i$, Send $Token (col_T, w_T)$ to $j$ \\[5pt]
			\STATE $father_i\rcp j$
		\ENDIF \label{token-joinFin}
		\vspace*{5pt}
	\ELSIF{$col_i=-1$ \label{token-falseDeb}}
			\STATE  Send $Token (col_T, w_T)$ to $w_T[1]$ \label{token-falseFin}\\[5pt]
		\STATE \texttt{//------ Case 2: i is blue and receives a red token} \label{token-case2Deb}\\[5pt]
	\ELSE
		\vspace*{5pt}
		\STATE \texttt{// If the red cluster is too small and under some asymmetric assumptions, 
			$i$ can dissolve it\\
			// Otherwise $i$ sends back the token to its sender}\\
		\IF{$Nb\_Identities(w_T)<m \land (Nb\_Identities(w_i)\geq m \lor col_i > col_T)$ \label{token-dissDeb}}
			\STATE ${\cal A} \rcp Build\_Tree(w_T)$
			\STATE  Send $Dissolution({\cal A})$ to $w_T[1]$ \label{token-dissFin}
		\ELSE 
			\STATE  Send $Token (col_T, w_T)$ to $w_T[1]$ \label{token-sendBack}
		\ENDIF 		
	\ENDIF \label{token-case2Fin}
\end{algorithmic}
\caption{On reception of $Token(col_T, w_T)$
\label{algo-tokenM}}
}
\end{algorithm}

When detecting that it can no longer communicate with its father, $i$ initiates a wave on the subtree rooted in it. $Delete$ messages are sent to all its neighbors, but only those that consider $i$ as their father take it into account. They set themselves free, and send $Delete$ messages to their neighbors, until all of $i$'s descendants are free. A node always knows its current father, so that the nodes that are set free are the current descendants of $i$. Thus, the set of nodes that are freed in this process is a subtree of the token's tree, and its complement is itself a subtree. If a freed node is visited by the token, then it is recruited, and the tree is modified according to the node to which it sends the token. Thus, the set of nodes that are free, but appear in the token, is always a subtree of the token's tree, and its complement is always a subtree too. When the token visits $i$'s father, it is corrected, and is consistent with the data present on the nodes.

\begin{algorithm}
\begin{algorithmic}[1]
\IF{$j=father_i$}
	\STATE $col_i\rcp null$
	\STATE $father_i\rcp null$
	\STATE send $Delete$ to all neighbors
\ENDIF
\end{algorithmic}
\caption{On a disconnection of node $j$}
\end{algorithm}

\begin{algorithm}
\begin{algorithmic}[1]
\IF{$j=father_i$}
	\STATE $col_i\rcp null$
	\STATE $father_i\rcp null$
	\STATE send $Delete$ to all neighbors
\ENDIF
\end{algorithmic}
\caption{On reception  of $Delete$ on node $i$ from a node $j$}
\end{algorithm}

If a dissolution occurs on a cluster that has been subject to a disconnection, the $Dissolution$ PIF is triggered on the subtree of all nodes that are in the token, and are reachable. The remaining of the cluster has already been set free, so that the cluster is eventually dissolved.

\begin{algorithm}[H]
{\small
\begin{algorithmic}[1]
	\IF{$k=father_i$}
		\STATE $col_i \rcp -1$
		\STATE $w_i \rcp \varepsilon $
	\ENDIF
	\vspace*{5pt}
	\IF{$|My\_Sons({\cal A})| > 0$}
		\STATE $\forall j \in My\_Sons({\cal A})\cap N(i) :$ Send $Dissolution({\cal A})$ to $j$
		\STATE $nbFeedbackDiss_i \rcp 0$
	\ELSE
		\STATE Send $FeedbackDiss({\cal A})$ to $My\_Father({\cal A})$
		\STATE $col_i \rcp null$
		\STATE $father_i\rcp null$
	\ENDIF
\end{algorithmic}
\caption{On the reception of $Dissolution({\cal A})$ from $k$
\label{algo-dissM}}
}
\end{algorithm}

\begin{algorithm}[H]
{\small
\begin{algorithmic}[1]
	\STATE $nbFeedbackDiss_i ++$\\
	\texttt{// if $i$ receives the last feedback it was waiting for}
	\IF{$nbFeedbackDiss_i = |My\_Sons({\cal A})\cap N(i)|$}
		\STATE send $FeedbackDiv({\cal A})$ to $My\_Father({\cal A})$
		\STATE $col_i \rcp null$
		\STATE $father_i\rcp null$
	\ENDIF
\end{algorithmic}
\caption{On the reception of $FeedbackDiss({\cal A})$
\label{algo-feedbackDissM}}
}
\end{algorithm}

If a division along an edge $(a, b)$ occurs in a cluster that has been subject to the disconnection of a link $(i, j)$, three cases can occur:\begin{enumerate}
\item $(i, j)$ is lower in the tree than $(a, b)$;
\item $(i, j)=(a, b)$;
\item $(i, j)$ is higher than $(a, b)$.
\end{enumerate}
In all these cases, the division PIF is propagated through the subtree of all the nodes that are still connected through this tree. In the first case, one of the two clusters is correct, and the other is not (its token \emph{``believes''} that some node are in the cluster, while they are not). In the second case, only one clusters exists, and it is correct. In the last case, only one cluster exists, and it is not correct (the token \emph{``believes''} that some node are in the cluster, while they are not).

\begin{algorithm}[H]
{\small
\begin{algorithmic}[1]
\IF{$k=father_i$}
	\IF{$i \in w_1$ \label{div-joinDeb}}
		\STATE $w_i\rcp w_1$
		\STATE $col_i\rcp w_1[1]$
	\ELSE
		\STATE $w_i\rcp w_2$
		\STATE $col_i\rcp w_2[1]$
	\ENDIF \label{div-joinFin}
	\vspace*{5pt}
	\IF{$|My\_Sons({\cal A})| > 0$ \label{div-pifDeb}}
		\STATE $\forall j\in My\_Sons({\cal A})\cap N(i)$ : Send $Division({\cal A},w_1,w_2)$ to $j$
		\STATE $nbFeedbackDiv_i \rcp 0$
	\ELSE
		\STATE Send $FeedbackDiv({\cal A},w_1,w_2)$ to $My\_Father({\cal A})$
	\ENDIF \label{div-pifFin}
\ELSE
	\IF{$|My\_Sons({\cal A})| > 0$}
		\STATE $\forall j \in My\_Sons({\cal A})\cap N(i) :$ Send $Dissolution({\cal A})$ to $j$
		\STATE $nbFeedbackDiss_i \rcp 0$
	\ELSE
		\STATE Send $FeedbackDiv({\cal A},w_1,w_2)$ to $My\_Father({\cal A})$
	\ENDIF
\ENDIF
\end{algorithmic}
\caption{On the reception of $Division({\cal A},w_1,w_2)$ from $k$
\label{algo-divM}}
}
\end{algorithm}

\begin{algorithm}[H]
{\small
\begin{algorithmic}[1]
	\STATE $nbFeedbackDiv_i ++$\\
	\texttt{// if $i$ receives the last feedback it was waiting for}
	\IF{$(i\in w_1 \land nbFeedbackDiv_i = |My\_Sons(w_1)\cap N(i)|) \lor (i\in w_2 \land nbFeedbackDiv_i = |My\_Sons(w_2)\cap N(i)|)$}
		\IF{$i=w_1[1]$}
			\STATE \texttt{// if $i$ is the root of the first tree}
			\STATE Random choice of $j \in N_i$, send $Token(i, w_1)$ to $j$
			\STATE $father_i\rcp j$
		\ELSIF{$i = w_2[1]$}
			\STATE \texttt{// if $i$ is the root of the second tree}
			\STATE Random choice of $j \in N_i$, send $Token(i, w_2)$ to $j$
			\STATE $father_i\rcp j$
		\ELSE
			\STATE send $FeedbackDiv({\cal A}, w_1, w_2)$ to $My\_Father({\cal A})$
		\ENDIF
	\ENDIF
\end{algorithmic}
\caption{On the reception of $FeedbackDiv({\cal A}, w_1, w_2)$
\label{algo-feedbackDivM}}
}
\end{algorithm}

First we prove that any node has information about its father in the tree borne by the token.

\begin{lemma}
If $i$ does not own the token, $(i, father_i)$ is an edge of the tree of a token of color $col_i$. 
\end{lemma}
\begin{proof}
$father_i$ is the last node to which $i$ has sent the token of color $col_i$ (algorithm 11, line 7; algorithm 12, line 8 and 27; algorithm 18, lines 6 and 10; algorithm 12, line 38 is when $i$ sends back a token that is not of its color; algorithm 13 and 14 show that if $i$ has no color anymore, it has no father anymore either).

The father of $i$ in the tree computed from a word $w$ is the node $j$ such that the first occurrence of $i$ in $w$ is preceded by $j$. Now, the first occurrence of $i$ in $w$ was written in $w$ the last time the token visited $i$. Indeed, no reduction of $w$ can delete the first occurrence of a node. Then, $i$ sent the token to a node, that added its id at the beginning of the token, \emph{ie} just before $i$'s id. This has never been removed, since $i$'s id would have been removed too in this case. Thus, the node to which $i$ sent the token is $j$. The father of $i$ in the tree of the token is $father_i$.
\end{proof}

We now focus on a link $(i, j)$ disappearing, with $father_i=j$, in a cluster of color $c$.

\begin{lemma}
When a link $(i, j)$ disappears, with $father_i=j$, eventually, all descendants of $i$ are eventually set free.
\end{lemma}
\begin{proof}
\begin{figure}[h]
\begin{center}\includegraphics[width=.5\linewidth]{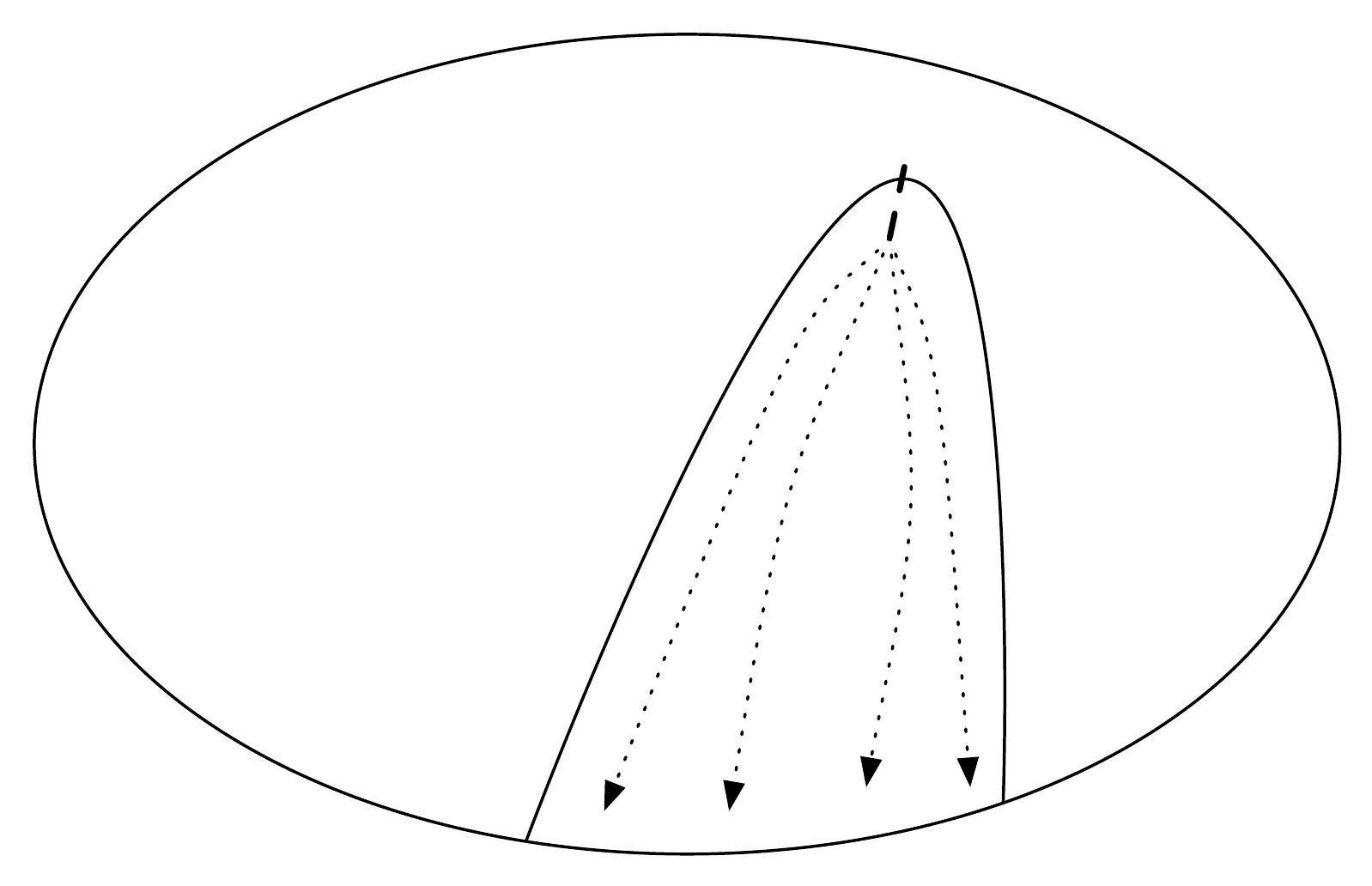}\end{center}
\caption{Descendants of $i$ are set free}
\end{figure}
When $i$ detects that it is no longer connected to its father, it sends a $Delete$ message to all its neighbors (algorithm 13). When a node receives a $Delete$ message, if it comes from its father, it sets itself free, and forwards it to all its neighbors (algorithm 14). Thus, all descendants of $i$ that can be reached through the spanning tree are set free. Now, if a descendant of $i$ cannot be reached through the spanning tree, a link in the subtree rooted in $i$ has also disappeared. Thus, a $Delete$ wave has triggered on a subtree of the spanning tree, lower than $i$, that has reached this node and set it free.
\end{proof}

Then, three case can happen:\begin{itemize}
\item a dissolution of the cluster occurs before the token of color $c$ hits $j$;
\item a division of the cluster occurs before the token of color $c$ hits $j$;
\item the token of color $c$ reaches $j$ before a dissolution or a division occurs.
\end{itemize}

Basically, when the token of color $c$ hits $j$, $j$ corrects the token to make it in accordance with the disconnection of $(i, j)$. A dissolution also leads to a correct configuration, and a division transfers the problem to the new cluster in which the disconnected subtree is included. 

Note that a version number in the token ensures that no two tokens (and thus, no two clusters) can have the same color.

%However, if the initiator of the cluster, the site the id of which is $c$, is in the subtree rooted in $i$, then it is set free, and it may create a new cluster of color $c$ (we will speak of two clusters in the following, even if, according to the definition of a cluster, there is only one; by the \emph{``old cluster''} we mean sites that have last received the first created token, and by the \emph{new cluster}, we mean sites that have last received the token that was created the most recently). 

%\begin{figure}[h]
%\begin{center}\includegraphics[width=.5\linewidth]{cas1'}\end{center}
%\caption{Site $c$ creates a new cluster, with the same color}
%\end{figure}

%If the initiator of the was on the other side of a removed edge than the token, it can initiate a new clustering process with the same color. The new token can recruit nodes that were already of this color. During the dissolution process, these nodes must be spared, so that no new inconsistency is created in the new token. 

\begin{lemma}
If the token hits $j$, the token is corrected.
\end{lemma}
\begin{proof}
On the token hitting $j$, $j$ executes algorithm 12, lines 4 and 5, which corrects the token. 

\begin{figure}[h]
\begin{center}\includegraphics[width=.5\linewidth]{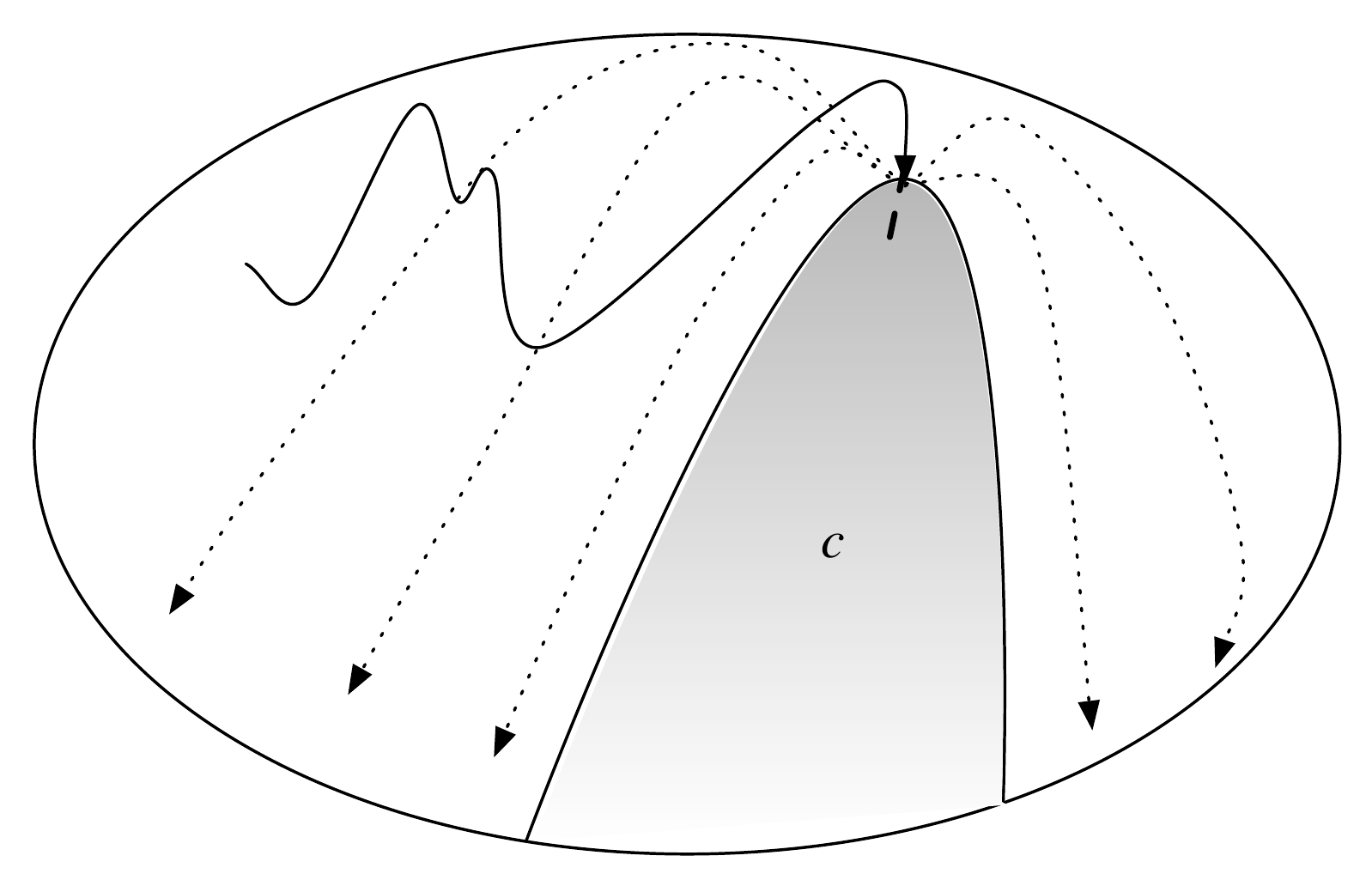}\end{center}
%\begin{center}\includegraphics[width=.5\linewidth]{cas3}\end{center}
\caption{The dissolution wave when the token hits $j$}
\end{figure}
\end{proof}

\begin{lemma}
The dissolution process on a tree $\mathcal A$ makes all nodes in $\mathcal A$ free, even in case of a disconnection of a link $(i, j)$ in $\mathcal A$ with $father_i=j$, except for nodes that have been recruited to other clusters.
\end{lemma}
\begin{proof}
\begin{figure}[h]
\begin{center}\includegraphics[width=.5\linewidth]{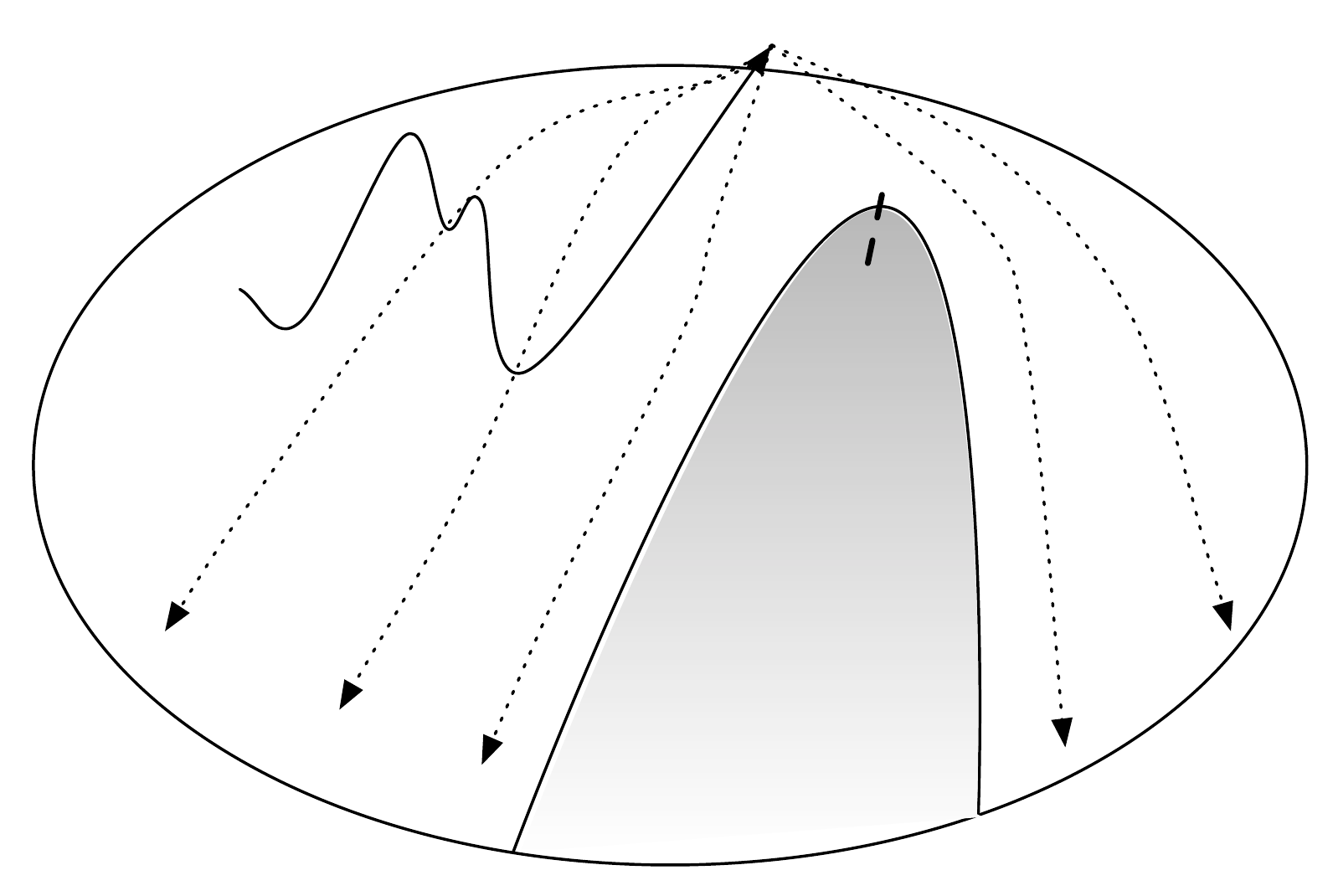}\end{center}
\caption{The dissolution wave}
\end{figure}
%When receiving a $Dissolution$ message a site forwards it to any of its sons, so that all reachable sites in $\mathcal A$ are hit. Then, it checks if this message comes from its father. If it is the case, it leaves the cluster and gets free. Thus, any node leaves the cluster if $father_i$ is actually its father in $\mathcal A$, \emph{ie} unless it has been recruited to another cluster meanwhile.
When receiving a $Dissolution$ message a node forwards it to all of its reachable sons, so that all reachable nodes in $\mathcal A$ are hit. Then, it leaves the cluster and gets free.

Nodes that are not reachable have received $Delete$ messages (see lemma above), and are free, unless they have been recruited by another cluster.
\end{proof}

\begin{lemma}
After a division process on a tree $\mathcal A$ along an edge $(a, b)$ with colors $c$ and $c'$, and with the disconnection of $(i, j)$, with $father_i=j$, nodes that have been recruited to another cluster keep their colors; nodes that are in the same connected component as $a$ take the color $c$; nodes that are in the same connected component as $b$ take the color $c'$; nodes that are disconnected from both $a$ and $b$ (due to the loss of $(i, j)$) are free, unless they have been recruited to another cluster.
\end{lemma}
\begin{proof}
The division wave is propagated along $\mathcal A$ minus the subtree rooted in $i$. All nodes in this subtree have set themselves free.
\begin{figure}[h]
\begin{center}\includegraphics[width=.5\linewidth]{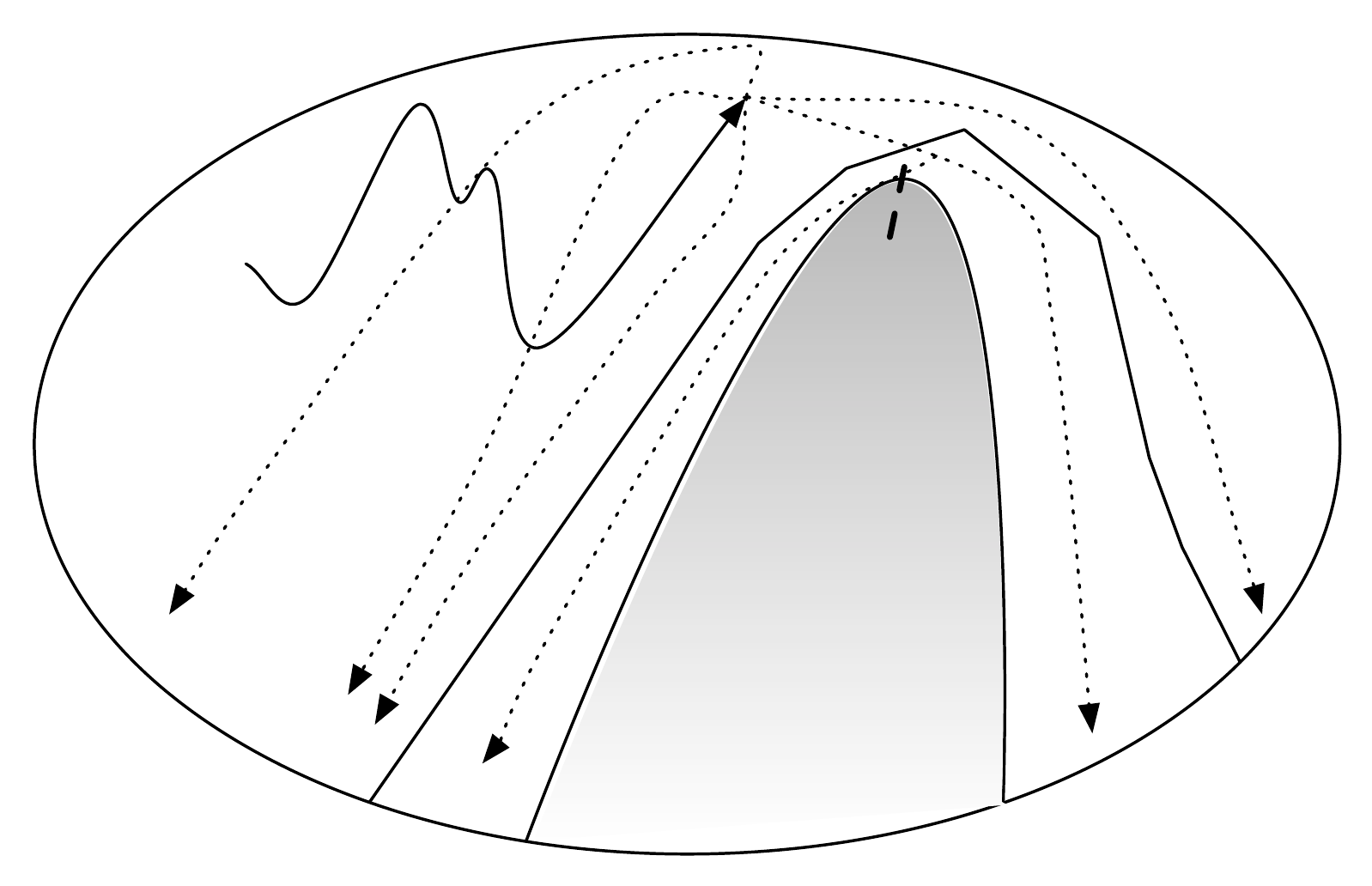}\end{center}
\caption{The division wave}
\end{figure}

Three cases occur:\begin{itemize}
\item $(i, j)$ is lower in $\mathcal A$ than $(a, b)$: the subtree rooted in $i$ is a subtree of the subtree rooted in $b$;
\item $(i, j)=(a, b)$: then all nodes in the subtree obtained by removing the subtree rooted in $b$ set their color to $c$; the others are either free or recruited by other clusters;
\item in other cases: all nodes in the subtree obtained by removing the subtree rooted in $i$ set their color to $c$; the others are either free or recruited by other clusters.
\end{itemize}

In all of these three cases, at most one cluster is not correct after the division. Other nodes have their expected colors, or are free.

\end{proof}

%\begin{definition}
%We call inconsistency:\begin{itemize}
%%\item the fact that node $i$'s color is $c\neq null$, and $i$ is not in the word of the token of color $c$, and only in this word; 
%\item the fact that $i$ is free, and occurs in a circulating word;
%\item the fact that there is no path from a node $i$ of color $c$ to the root of the spanning tree of $V_c$;
%\item a message $Delete$ circulates on a link $(i, j)$, with $father_j=i$.
%\end{itemize}
%\end{definition}
%
%
%
%
%
%
%
%
%
%
%
%
%
%
%
%
%\begin{lemma}
%From any configuration obtained from a correct configuration by adding and removing some nodes and links, the number of inconsistencies decreases to 0 \emph{whp}.
%\end{lemma}
\begin{lemma}\label{uniqueM}
Eventually, a token of color $c$ exists if and only if $V_c\neq\emptyset$. In this case, this token is unique and is in $N({V_c})$.
\end{lemma}
\begin{proof}
Suppose a token of color $c$ exists, and is on a node $i$. Then $i$ either has created it, or has received it from a node $j$. If it has created it, $col_i=c$, and $i\in V_c$. Thus, $V_c\neq\emptyset$. If it has received it from $j$, and $i$ is not of the color $c$, then it sends it back to $j$ (algorithm 12, line 33) or it destroys both the token and $V_c$ (algorithm 12 line 31). If the link $(i, j)$ disconnects before $j$ has had the time to send it back to $i$, then $i$ deletes its cluster (since $father_i=j$: algorithm 13 and lemma 10) and $j$ deletes the token (algorithm 12, lines 25 and 33: $j$ is unable to send back the token, and does nothing, so that the token disappears). Thus, $i\in V_C$, and $j\in N(i)$. In any case, if a token of color $c$ exists, then $V_c\neq \emptyset$, and the token is in $N(V_c)$. 

The only token creation is on a free node executing algorithm \ref{algo-create}. The color of this token is constituted of the id of the free node that creates it, and of a unique version number. Thus, no two tokens can have the same color.

Now, suppose $V_c\neq\emptyset$. Let $i\in V_c$. Since $i$ is of color $c$, it has received a token of color $c$. Now, the only time a token of color $c$ disappears is when a dissolution or a division process is launched: in algorithm 12, each time a node receives a $Token$ message, it sends a $Token$ message or triggers a division or a dissolution. When a division is launched, all sites in $V_c$ change their color (lemma 13), so that eventually $V_c=\emptyset$. When a dissolution is triggered, all sites in $V_c$ are set free (lemma 12), and eventually $V_c=\emptyset$.
\end{proof}

\begin{theorem}
From any configuration obtained from a correct configuration by adding and removing some nodes and links, the algorithm converges to a correct configuration \emph{whp} (and then, to a legitimate configuration).
\end{theorem}
\begin{proof}
The addition of a link or a node (that executes its initialization procedure) to a correct configuration leads to a correct configuration (see definition of a correct configuration).

If a node is removed, then we treat it as if all its adjacent links were removed.

So, consider the disconnection of a link $(i, j)$. If $i$ and $j$ have no father-son relation, nothing happens (algorithm 13). Assume $j$ is the father of $i$. Then $i$ sets itself free and sends a $Delete$ message to all its neighbors (algorithm 13). All nodes descending from $i$ in $\mathcal A$ are free (lemma  10). %Its sons set themselves free and send $Delete$ messages to their own neighbors. By induction, all descendants of $i$ that can be reached set themselves free. If descendant of $i$ cannot be reached, it means that an edge in the subtree rooted in $i$ has also been disconnected. Then, the descendant of this edge have already been set free. In any case, all descendants of $i$ are set free.
%According to lemma ??

Then one of the three following events happens first:\begin{itemize}
\item the token hits $j$;
\item the cluster is dissolved;
\item the cluster is divided.
\end{itemize}

If the token hits $j$, lemma 11 shows that the cluster and its token are corrected. If the cluster is dissolved, then all nodes are free (lemma 12) , and this is correct. If the cluster is divided, the new cluster that inherited the edge $(i, j)$ is still incorrect, and the other cluster either does not exist, or is correct (lemma 13). In the latter case, the same applies. Then, by induction, either a token of $j$'s color eventually visits $j$, or $j$'s cluster is eventually dissolved \emph{whp}. Thus, \emph{whp}, the clusters containing $j$ is correct.

The same applies to all clusters, and when all clusters are correct, then the configuration is correct.

%Now if $j$ receives the token that contains the link $(i, j)$, it corrects it (algorithm 12, lines 2 to 6). If $j$ never receives this token, it means that $j$ has left the cluster, either due to a disconnection, or due to a dissolution. If it is due to a disconnection, then a subtree containing $j$, and therefore containing the subtree rooted in $i$, has been disconnected, and some node higher than $j$ in the tree will correct the token (in the worst case, the node that owns the token sees that it has no more neighbor, and corrects it). If the cluster has been dissolved, no trace of the topological modification are kept in any token.
\end{proof}

\begin{theorem}
Starting from a legitimate configuration, and adding a topological modification, the algorithm converges to a correct configuration and the only clusters modified are at worst the cluster in which the modification took place, and the adjacent clusters.
\end{theorem}
\begin{proof}
A stable non-divisible cluster can be modified by the algorithm only if a neighboring node is free (algorithm 12). 

The starting configuration is a legitimate one, with the addition of a topological modification. If this topological modification is the addition of a link, nothing happens. If this is the addition of a node, then, this node will be recruited by a cluster (being surrounded by stable clusters, if it creates its own cluster, at the first step, its token reaches a stable cluster and its cluster is dissolved), and, at worst it will trigger a division. Other clusters are stable and non-divisible, and have no neighboring free node : they are not modified.

Now, consider a link disconnection. If this link was between two nodes without father-son relationship, then nothing happens. Thus, let consider the case when a link $(i, j)$ disappears, and $i$ is the son of $j$ in cluster $V_c$. Then, $i$ sets all its descendant free (lemma 11). Neighboring clusters may recruit them. Nodes in the cluster that are not descendant of $i$ remain in the cluster. At some point, the token reaches $j$ and is corrected. If the cluster is still stable, then it goes on. Otherwise, its node may be set free by its dissolution. Only nodes in $V_c$ may be set free. Indeed, all other nodes are in stable clusters, and remain in stable cluster according to lemma \ref{stablenode}. Since only nodes in $V_c$ may be set free, and all other clusters are stable and non-divisible, only clusters that have a neighbor in $V_c$ can be modified.
\end{proof}

\section{Conclusion}

The algorithm presented in this paper computes a distributed clustering on arbitrary topologies in a totally decentralized way, and tolerates mobility. Most of the distributed clustering algorithms so far are based on the election of a leader, called a \emph{``clusterhead''}. Most of them suffer from the fact that a single link or node disconnection can entail a complete change of the clustering. Moreover, most of them specify the clustering problem as finding disjoint clusters that are star subgraphs. The specification we use is more advanced: we build clusters with a size higher than a parameter $m$, and are locally optimal, in the sense that no cluster can be divided into two clusters greater than $m$ (no global optimum can be computed for such a problem without a global view of the system, which would be in opposition to the distributed nature of the algorithm). The way we handle mobility ensures that the only clusters affected by a node or link disconnection are, in the worst case, the clusters in which it took place, and the clusters adjacent to it. The reconfiguration is, as far as possible, local. Indeed, the loss of the $m^\text{th}$ node in a cluster needs this cluster to be deleted (to fit the $|V|\geq m$ constraint), and adjacent clusters to recruit the newly orphan nodes.

Thus, this algorithm provides a locally optimal clustering, and, in terms of affected nodes, an optimal reconfiguration. We now aim at better studying the complexity of this algorithm, both on the theoretical level (although we already know that the size of the network has only a weak influence on the time to obtain a global clustering, thanks to the concurrent construction of the different clusters), and through simulations.

We are also interested in making this algorithm self-stabilizing, in order to take into account, for instance, the possible message losses. Starting from any arbitrary configuration, the values of the nodes variables being arbitrary, and arbitrary messages being in transit, the system has to reach a clustering meeting the specification. We already worked on a self-stabilizing random walk based token circulation (\cite{BeBF04b}), so that we have clues on how to manage failures on the token. We still have to use this block to compute a self-stabilizing distributed clustering.

\bibliographystyle{alpha}
\bibliography{biblio}

\end{document}